\documentclass[onecolumn,draft]{IEEEtran}
\ifCLASSINFOpdf
\else
\fi
\usepackage[cmex10]{amsmath}
\usepackage{amsthm}
\usepackage{bm}
\usepackage[final]{graphicx}
\newtheorem{theorem}{Theorem}
\newtheorem{dfn}{Definition}
\begin{document}
%
\title{On secure network coding  with uniform  wiretap sets}
%
%
%
\author{Wentao~Huang,
        Tracey~Ho,~\IEEEmembership{Senior~Member,~IEEE,}
        Michael~Langberg,~\IEEEmembership{Member,~IEEE,}
        and~Joerg~Kliewer,~\IEEEmembership{Senior~Member,~IEEE,}
\thanks{Wentao~Huang and Tracey~Ho are with the Department of Electrical Engineering, California Institute of Technology, Pasadena, CA
91125, USA. Email:\{whuang,tho\}@caltech.edu}
\thanks{Michael~Langberg is with the Department of Mathematics and Computer Science at The Open University of Israel, 108 Ravutski St., Raanana
43107, Israel. Email:mikel@openu.ac.il}
\thanks{Joerg Kliewer is with the Klipsch School of Electrical and Computer Engineering, New Mexico State University, Las Cruces,
NM 88003, USA. Email: jkliewer@nmsu.edu}}

\maketitle

\begin{abstract}
This paper shows determining the secrecy capacity of a unicast network with uniform wiretap sets is at least as difficult as the $k$-unicast problem. In particular, we show that a general $k$-unicast problem can be reduced to the problem of finding the secrecy capacity of a corresponding single unicast  network with uniform link capacities and one arbitrary wiretap link.
\end{abstract}


%
\IEEEpeerreviewmaketitle

\section{Introduction}The secure network coding problem, introduced by  Cai and Yeung \cite{Cai2002}, concerns information theoretically secure communication over a network where an unknown subset of network links may be wiretapped. A secure code ensures that the wiretapper
obtains no information about the secure message being communicated.
The secrecy capacity of a network, with respect to a given collection of possible wiretap sets,
 is the maximum rate of communication such that for any one of the wiretap sets the secrecy constraints  are satisfied. Types of secrecy constraints studied in the literature include perfect secrecy, strong secrecy and weak secrecy.

This paper considers the problem of finding the secrecy capacity of a network when we allow network nodes in addition to the source to generate independent randomness. We show that a general $k$-unicast problem can be
reduced to a corresponding single unicast secrecy capacity problem with uniform link capacities where any single link can be wiretapped. This implies that determining the secrecy capacity, even in the simple case of a single unicast, uniform link capacities and unrestricted wiretap sets where any single
link can be wiretapped, is at least as
difficult as the long standing open problem of determining the capacity region of multiple-unicast network coding, which is not
presently known to be in P, NP or undecidable~\cite{Langberg2009}. In contrast, under the assumption that only the source can generate randomness, link capacities are uniform and up to $z$ arbitrary links can be wiretapped, Cai and Yeung \cite{Cai2002} showed
that the secrecy capacity is given by the cut-set bound and linear codes suffice to achieve capacity. The secure network coding problem with restricted wiretap sets and/or non-uniform link capacities has been considered by Cui et al.~\cite{cui12}, who studied achievable coding schemes, and by Chan and Grant~\cite{chan_grant_08}, who showed that determining multicast secrecy capacity with restricted (non-uniform) wiretap sets is at least as difficult as determining capacity for multiple-unicast network coding.
Our reduction follows the same core ideas appearing in \cite{chan_grant_08} with two
differences. First, by introducing the idea of key cancellation and
replacement at intermediate nodes, our construction does not need to impose
restrictions on which links can be wiretapped. Secondly, unlike the
reduction in [4] which involves multiple terminals, ours only needs a single
destination. Therefore, our construction shows that even a single secure
unicast in the uniform setting (equal capacity links where any one of which
can be wiretapped) is as difficult as a k-unicast problem.

\section{Model}
A network is represented by a directed graph $\mathcal{G}=(\mathcal{V},\mathcal{E})$, where $\mathcal{V}$ is the set of vertices which represent nodes, and $\mathcal{E}$ 
is the set of edges that represent links.  We assume links have equal capacity and there may be multiple links between a pair of nodes. There is a source node $S \in \mathcal{V}$ and a destination node $D \in \mathcal{V}$.  The source wants to communicate a message $M$, uniformly drawn from a finite alphabet set $\mathcal{S}_n$, to the destination using a code with length or duration $n$. Then the rate of the code is $n^{-1}\log |\mathcal{S}_n|$. For the network coding problem, we say that a communication rate $R$ is feasible if there exists a sequence of length-$n$ codes such that $|\mathcal{S}_n|=2^{nR}$ and the probability of decoding error tends to 0 as $n \to \infty$.

For the secure network coding problem,  we specify additionally a collection $\mathcal{A}$ of wiretap link sets, \emph{i.e.}, $\mathcal{A}$ is a collection of subsets of $\mathcal{E}$ such that an eavesdropper can wiretap any one set in $\mathcal{A}$.  We consider three kinds of secrecy constraints:  the requirement, for all $A \in \mathcal{A}$, that
\begin{align}
  I(M;X^n(A)) =  0\label{perfect}\end{align}corresponds to perfect secrecy, that
  \begin{align}
  I(M;X^n(A)) \to 0 \text{ as } n \to \infty\label{strong}\end{align}corresponds to strong secrecy, and that
  \begin{align}
  \frac{I(M;X^n(A))}{n}  \to 0 \text{ as } n \to \infty\label{weak}
\end{align}corresponds to weak secrecy,
where $X(A) = \{X(a,b): (a,b) \in A \}$, and $X(a,b)$ is the signal transmitted on the link $(a,b)$.
We say a secrecy rate $R$ is feasible if the communication rate $R$ is feasible and the prescribed secrecy condition is satisfied. The secrecy capacity of the network is defined as the supremum of all feasible secrecy rates.
\section{Main Result}
\begin{figure}[h!]
  \begin{center}
      \includegraphics[width=0.7\textwidth]{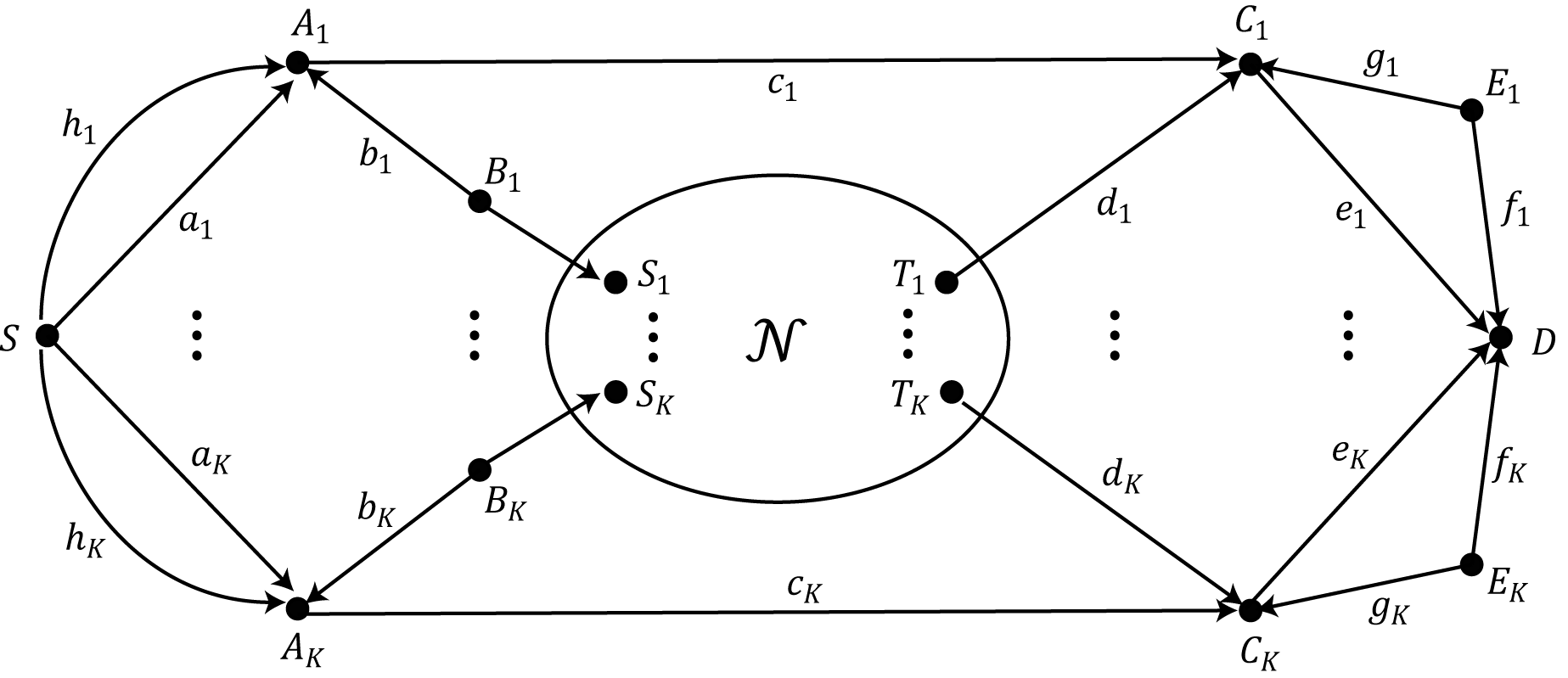}
  \caption{The source $S$ wants to communicate with the destination $D$ secretly (with either weak secrecy, strong secrecy or perfect secrecy). $\mathcal{N}$ is an embedded general network. Links are labeled by the signals transmitted on them.}\label{secmul}
           \end{center}
\end{figure}
\begin{theorem}\label{th1}Given any $K$-unicast problem with source-destination pairs $\{(S_i,T_i),\ i=1,...,K\}$ of unit rate, the corresponding secure communication problem in Figure \ref{secmul} with unit capacity links, any one of which can be wiretapped, has secrecy capacity $K$ (under perfect, strong or weak secrecy requirements) if and only if the  $K$-unicast problem is feasible.
\end{theorem}
\begin{proof}
  ``$\Rightarrow$''. Suppose a secrecy rate of $K$ is achieved by a code with length $n$. Let $M$ be the source input message, then $H(M)=Kn$. Because there is no shared randomness between different nodes, $M$ is independent with $\{ \bm{d}_1^n,\ \bm{f}_k^n,\ k=1,...,K \}$. Hence
  \begin{align}\label{Kn}
    H(M|\bm{d}_1^n,\bm{f}_2^n,...,\bm{f}_K^n)=Kn.
  \end{align}
  By the chain rule,
  \begin{align*}
   H(M|\bm{c}_1^n,\bm{d}_1^n,\bm{f}_2^n,...,\bm{f}_K^n) + H(\bm{c}_1^n|\bm{d}_1^n,\bm{f}_2^n,...,\bm{f}_K^n) = H(M,\bm{c}_1^n|\bm{d}_1^n,\bm{f}_2^n,...,\bm{f}_K^n) \ge H(M|\bm{d}_1^n,\bm{f}_2^n,...,\bm{f}_K^n).
  \end{align*}
  So
  \begin{align}\label{first}
    H(M|\bm{c}_1^n,\bm{d}_1^n,\bm{f}_2^n,...,\bm{f}_K^n) \ge H(M|\bm{d}_1^n,\bm{f}_2^n,...,\bm{f}_K^n) - H(\bm{c}_1^n|\bm{d}_1^n,\bm{f}_2^n,...,\bm{f}_K^n) \ge (K-1)n,
  \end{align}
  where the last inequality holds because of (\ref{Kn}) and $H(\bm{c}_1^n|\bm{d}_1^n,\bm{f}_2^n,...,\bm{f}_K^n) \le H(\bm{c}_1^n) \le n$. Similarly,
  \begin{align}
    H(M|\bm{c}_1^n,\bm{d}_1^n,\bm{f}_2^n,...,\bm{f}_K^n) &\le H(M|\bm{c}_1^n,\bm{d}_1^n,\bm{f}_2^n,...,\bm{f}_K^n,\bm{e}_2^n,...,\bm{e}_K^n) + H(\bm{e}_2^n,...,\bm{e}_K^n|\bm{c}_1^n,\bm{d}_1^n,\bm{f}_2^n,...,\bm{f}_K^n)\nonumber\\
    & \le   n\epsilon_n + H(\bm{e}_2^n,...,\bm{e}_K^n|\bm{c}_1^n,\bm{d}_1^n,\bm{f}_2^n,...,\bm{f}_K^n)\label{fano}\\
    & \le  n\epsilon_n + (K-1)n \label{Mcond},
  \end{align}
  where $\epsilon_n \to 0$ as $n \to \infty$ and (\ref{fano}) is due to the cut set $\{\bm{c}_1^n, \bm{d_1}^n, \bm{f}_2^n,...,\bm{f}_K^n\}$ from $S$ to $D$ and Fano's inequality. Hence it follows
   \begin{align}
     H(\bm{c}_1^n) & \ge H(\bm{c}_1^n|\bm{d}_1^n,\bm{f}_2^n,...,\bm{f}_K^n) \nonumber \\
     & \ge H(M|\bm{d}_1^n,\bm{f}_2^n,...,\bm{f}_K^n) -  H(M|\bm{c}_1^n,\bm{d}_1^n,\bm{f}_2^n,...,\bm{f}_K^n)\label{c11}\\ &\ge n -n\epsilon_n\label{c12},
   \end{align}
   where (\ref{c11}) holds because of (\ref{first}), and (\ref{c12}) follows from (\ref{Kn}) and (\ref{Mcond}).
   Also notice that
  \begin{align}\label{Mcdf}
    H(M|\bm{c}_1^n,\bm{d}_1^n,\bm{f}_2^n,...,\bm{f}_K^n) \ge H(M|\bm{c}_1^n,\bm{f}_2^n,...,\bm{f}_K^n) - H(\bm{d}_1^n|\bm{c}_1^n,\bm{f}_2^n,...,\bm{f}_K^n),
  \end{align}
  where
  \begin{align}\label{Mc}
    H(M|\bm{c}_1^n,\bm{f}_2^n,...,\bm{f}_K^n) = H(M|\bm{c}_1^n) \ge Kn - n\delta_n,
  \end{align}
  with $\delta_n \to 0 \text{ as } n \to 0$. Here the first equality holds because $\{M, \bm{c}_1^n\}$ is independent with $\{\bm{f}_i^n,\ i=1,...,K\}$ and the second inequality holds due to the weak secrecy constraint. Note that all arguments extend naturally to the cases of strong and perfect secrecy because they are even stronger conditions. Therefore by (\ref{Mcond}), (\ref{Mcdf}) and (\ref{Mc}) we have
  \begin{align}
    H(\bm{d}_1^n) & \ge H(\bm{d}_1^n|\bm{c}_1^n,\bm{f}_2^n,...,\bm{f}_K^n)\nonumber\\
    & \ge H(M|\bm{c}_1^n,\bm{f}_2^n,...,\bm{f}_K^n) - H(M|\bm{c}_1^n,\bm{d}_1^n,\bm{f}_2^n,...,\bm{f}_K^n)\nonumber\\
    & \ge n-n\epsilon_n-n\delta_n.\label{hdn}
  \end{align}
   Furthermore, by the independency between the sets of $\{M, \bm{c}_1^n, \bm{d}_1^n\}$ and $\{\bm{f}_i^n,\ i=1,...,K\}$ we also have $$H(M|\bm{c}_1^n,\bm{d}_1^n,\bm{f}_2^n,...,\bm{f}_K^n)=H(M|\bm{c}_1^n,\bm{d}_1^n)$$. According to (\ref{first}) and (\ref{Mcond}), it is bounded by
  \begin{align}\label{conMbnd}
    (K-1)n \le H(M|\bm{c}_1^n,\bm{d}_1^n) \le  n\epsilon_n + (K-1)n.
  \end{align}

  Now consider the joint entropy of $M,\ \bm{d}_1^n,\ \bm{c}_1^n$ and expand it in two ways
  \begin{align*}
    H(M,\bm{d}_1^n,\bm{c}_1^n) & = H(\bm{c}_1^n|M,\bm{d}_1^n) + H(M|\bm{d}_1^n) + H(\bm{d}_1^n)\\
    & = H(M|\bm{c}_1^n,\bm{d}_1^n) + H(\bm{d}_1^n|\bm{c}_1^n)+H(\bm{c}_1^n) \le (K+1)n  + n\epsilon_n,
  \end{align*}
  where the last inequality holds because of (\ref{conMbnd}) and $H(\bm{d}_1^n|\bm{c}_1^n) \le n$, $H(\bm{c}_1^n) \le n$.
  Therefore
  \begin{align}\label{cMd}
    H(\bm{c}_1^n|M,\bm{d}_1^n) \le (K+1)n  + n\epsilon_n - H(M|\bm{d}_1^n) - H(\bm{d}_1^n) \le  2n\epsilon_n + n\delta_n,
  \end{align}
   where (\ref{hdn}) and $H(M|\bm{d}_1^n) = Kn$ (because $M$ and $\bm{d}_1^n$ are independent by construction) are used to establish the inequality. And so by observing the Markov chain $(M,\bm{d}_1^n) \to
  (M,\bm{b}_1^n) \to \bm{c}_1^n$, it follows
  \begin{align}\label{cMb}
      H(\bm{c}_1^n|M,\bm{b}_1^n) = H(\bm{c}_1^n|M,\bm{b}_1^n,\bm{d}_1^n) \le H(\bm{c}_1^n|M,\bm{d}_1^n) \le 2n\epsilon_n + n\delta_n.
  \end{align}

  Then consider the joint entropy of $M,\ \bm{b}_1^n,\ \bm{c}_1^n$ and expand it in two ways
  \begin{align*}
    H(M, \bm{b}_1^n, \bm{c}_1^n) &= H(\bm{b}_1^n|M,\bm{c}_1^n) + H(M|\bm{c}_1^n)+H(\bm{c}_1^n)\\
    &= H(\bm{c}^n_1|M, \bm{b}_1^n) + H(M|\bm{b}_1^n) + H(\bm{b}_1^n) \le (K+1)n + 2n\epsilon_n + n\delta_n,
  \end{align*}
  where the last inequality holds due to (\ref{cMb}) and $H(M|\bm{b}_1^n) = Kn$, $H(\bm{b}_1^n) \le n$.
  Therefore by (\ref{c12}) and the weak secrecy constraint $H(M|\bm{c}_1^n) \ge Kn - n\delta_n$, we have
  \begin{align}\label{bMc}
    H(\bm{b}_1^n|M,\bm{c}_1^n) \le (K+1)n  + 2n\epsilon_n + n\delta_n - H(M|\bm{c}_1^n) - H(\bm{c}_1^n) \le 3n\epsilon_n + 2n\delta_n.
  \end{align}
  So \begin{align*}
    H(\bm{b}_1^n|M,\bm{d}_1^n) &\le H(\bm{b}_1^n,\bm{c}_1^n|M,\bm{d}_1^n)\\
    &= H(\bm{b}_1^n|M,\bm{c}_1^n,\bm{d}_1^n) + H(\bm{c}_1^n|M, \bm{d}_1^n)\\
    & \le H(\bm{b}_1^n|M,\bm{c}_1^n) + H(\bm{c}_1^n|M, \bm{d}_1^n)\\
    &\le 3n\epsilon_n  + 2n\delta_n
    + 2n\epsilon_n + n\delta_n = 5n\epsilon_n + 3n\delta_n,
  \end{align*}
  where the last inequality invokes (\ref{bMc}) and (\ref{cMd}).
  Notice that $M$ is independent with $\{\bm{b}_1^n,\bm{d}_1^n\}$, so
  \begin{align}\label{bd}
    H(\bm{b}_1^n|\bm{d}_1^n)=H(\bm{b}_1^n|M,\bm{d}_1^n)\le 5n\epsilon_n + 3n\delta_n.
  \end{align}

  Now we bound the entropy of $\bm{b}^n_1$. Again consider the joint entropy,
  \begin{align*}
    H(M, \bm{b}_1^n, \bm{c}_1^n) &= H(\bm{c}^n_1|M, \bm{b}_1^n) + H(M|\bm{b}_1^n) + H(\bm{b}_1^n)\\
    &= H(\bm{b}_1^n|M,\bm{c}_1^n) + H(M|\bm{c}_1^n)+H(\bm{c}_1^n) \ge (K+1)n  - n\epsilon_n -n\delta_n,
  \end{align*}
  where the last inequality holds because of (\ref{c12}), the secrecy condition $H(M|\bm{c}_1^n) \ge Kn - n\delta_n$, and $H(\bm{b}_1^n|M,\bm{c}_1^n) \ge 0$.
  So by (\ref{cMb}) and because $H(M|\bm{b}_1^n) = Kn$, we have
  \begin{align}\label{Hb}
    H(\bm{b}_1^n) \ge (K+1)n  -n\epsilon_n  -n\delta_n - H(\bm{c}_1^n|M,\bm{b}_1^n) - H(M|\bm{b}_1^n) \ge n  -3n\epsilon_n -2n\delta_n.
  \end{align}

   Finally, by (\ref{bd}) and (\ref{Hb}),
  \begin{align}
     I(\bm{b}_1^n;\bm{d}_1^n)
     \ge H(\bm{b}_1^n) - H(\bm{b}_1^n|\bm{d}_1^n)
    & \ge n  -8n\epsilon_n -5n\delta_n \nonumber,
  \end{align}
  The  above argument extends to all other paths naturally (by renumbering the notations accordingly), so
    \begin{align}
     I(\bm{b}_i^n;\bm{d}_i^n) \ge n  -8n\epsilon_n -5n\delta_n \nonumber,\ \ \forall  i = 1,..., K.
  \end{align}
  Therefore $\forall  i = 1,..., K$, by the channel coding theorem, if we employ an outer code of length $n$ by encoding $\bm{b}_i^n$ as a supersymbol, then there exists an inner code that achieves a rate of $n - 8n\epsilon_n -8n\delta_n$ from $B_i$ to $T_i$, and so the overall rate is
  $$R_i  \ge \frac{n -8n\epsilon_n -5n\delta_n}{n} \to 1\ \ \text{as }n\ \to \infty.$$
  Because $B_i$ can be viewed as a virtual source of $S_i$, so $\forall  i = 1,..., K$, the unicast from node $S_i$ to $T_i$ of rate 1 is feasible.

  ``$\Leftarrow$''. The secrecy capacity is upper bounded by $K$ due to the min cut from $S$ to $D$. And secrecy rate $K$ is achieved by the scheme described in Figure \ref{schemefig}, i.e., let $\epsilon_n^{(i)}$ be the probability of error for the unicast from $S_i$ to $T_i$, then the probability of error from $S$ to $D$ is upper bounded by $K\epsilon^*_n \to 0$ as $n \to \infty$, where $\epsilon^*_n = \max_{i} \epsilon^{(i)}_n$. Note that the scheme  achieves perfect secrecy, which in turn implies strong and weak secrecy requirements are also satisfied.
\end{proof}
\begin{figure}[h!]
  \begin{center}
      \includegraphics[width=0.7\textwidth]{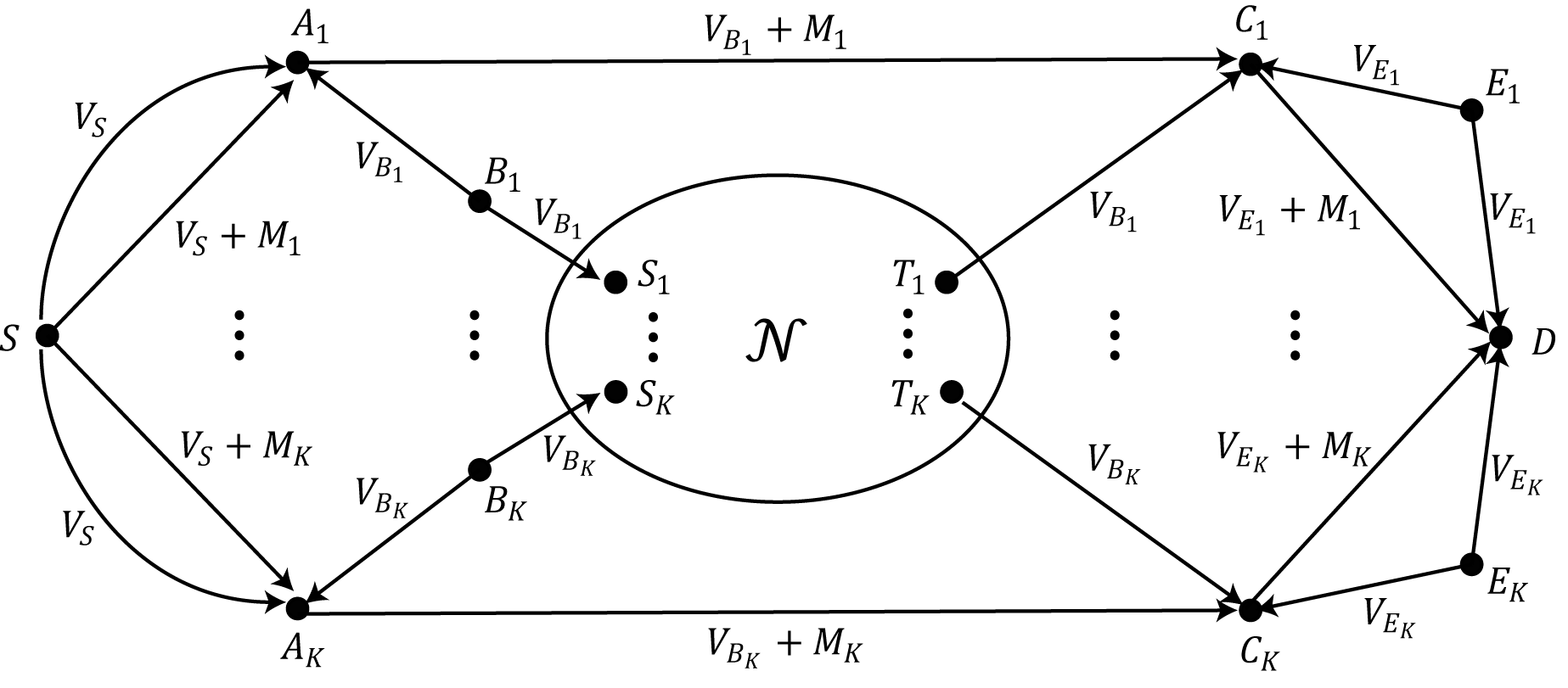}
  \caption{A scheme to achieve secrecy rate $K$. $V_x$ is the local key injected by node $x$, with $H(V_x)=1,\ \forall x$. $M_i,\ i=1,...,K$ are source input messages, with $H(M_i)=1,\ i=1,...,K$.}\label{schemefig}
           \end{center}
\end{figure}

The above result can be easily extended to the case of zero error communication and perfect secrecy.  In this case, we
say a rate $R$ is feasible if there exists a code with finite length
$n$ such that $|\mathcal{S}_n|=2^{nR}$ and the probability of decoding error is strictly zero.
Then for the secrecy communication problem in Figure \ref{secmul}, its zero error perfect secrecy capacity is $K$ if and only
if the $K$-unicast for source-destination pairs $\{(S_i,T_i),\ i=1,...,K\}$ of unit rate is feasible with zero error. The proof of this claim follows the same outline as the proof of Theorem \ref{th1}, with the difference that all $\epsilon_n$ and $\delta_n$ become strictly 0. Then (\ref{bd}) implies that $\bm{b}_1^n$ is a function of $\bm{d}_1^n$, and hence that it can be perfectly decoded from  $\bm{d}_1^n$.

Conversely, we note that for any given weakly secure communication problem where any one link can be wiretapped, we can  construct a corresponding  communication problem without security constraints (which can in turn be reduced to an equivalent multiple-unicast problem by~\cite{dougherty06nonreversibility}) that is feasible if and only if the secure communication problem is feasible. The equivalent communication problem is defined on a specialized version of the $A$-enhanced network in~\cite{dikaliotis12}, stated here in simplified form for convenience as follows.

\begin{dfn}
Consider a secure communication problem on a network $\mathcal{N}$ represented by a directed graph $\mathcal{G}=(\mathcal{V},\mathcal{E})$ with the collection of wiretap sets $A=\{\{e\}:\;e\in\mathcal{E}\}$ comprising individual links. Let $c_e$ denote the capacity of link $e\in\mathcal{E}$ and let $\mathcal{E}_{\text{out}}(i)$ denote the set of links $(i,j)$ originating at node $i\in\mathcal{V}$. The $A$-enhanced network $\mathcal{N}(A)$ on graph $\check{\mathcal{G}}=(\check{\mathcal{V}},\check{\mathcal{E}})$ is defined as follows:
\begin{enumerate}
\item For each link $e=(i,j)\in\mathcal{E}$, create an eavesdropper node $v_e$ and a node $u_{e}$, replace $(i,j)$ by two links $(i,u_{e})$ and $(u_e,j)$  and create a link $(u_e,v_e)$, all of capacity $c_e$.
    \item For each node $i\in\mathcal{V}$ create a message node $v_i$ and a random key node $\bar{v}_i$.\item Create an overall key node $v_L$.
    \item For each $i\in\mathcal{V}$, create a set $H_i$  of links  from node $v_i$ to all of the nodes in $\big\{i\big\}\cup\big\{v_e:e\in\mathcal{E}\big\}$, and a set $\bar{H}_i$  of links   from  node $\bar{v}_i$ to nodes $i$ and $v_L$, all of which have capacity
\begin{align*} \check{c}_i=\sum_{e\in\mathcal{E}_{\text{out}}(i)}c_e.
\end{align*}\item
 For each link $e\in\mathcal{E}$, create a link $(v_L,v_e)$ of capacity
\begin{align*}
\sum_{e'\in\mathcal{E}}c_{e'}-c_e.
\end{align*}
\item $\check{\mathcal{V}}=\mathcal{V}\cup\big\{v_i:i\in\mathcal{V}\big\}\cup\big\{\bar{v}_i:i\in\mathcal{V}\big\}\cup\big\{u_e,v_e:e\in\mathcal{E}
        \big\}\cup\{v_L\}$.
\item $\check{\mathcal{E}}=\bigcup_{i\in\mathcal{V}}(H_i\cup\bar{H}_i)\cup\big\{(i,u_{e})$, $(u_e,j)$, $(u_e,v_e):e=(i,j)\in\mathcal{E}\big\}\cup\big\{(v_L,v_e):e\in\mathcal{E}\big\}$.
\end{enumerate}
\label{dfn:noiseless_network_lower_bound}
\end{dfn}
The communication requirements in the A-enhanced network are as follows.
For each message that originates at a node $i$ in the original secure
communication problem, in the A-enhanced network, the message originates
instead at the corresponding message node $v_i$ and is demanded by the same
destination nodes as in the original problem. In addition,  the communication problem on the
$A$-enhanced network also requires a random key message $L_i\in\mathcal{L}_i=\{1,\ldots,2^{n\check{c}_i}\}$ to be delivered from each random key node $\bar{v}_i$ to all eavesdropper nodes $\{v_e:e\in\mathcal{E}\}$. Intuitively, if the communication problem on the
$A$-enhanced network is solved, then it implies that the information observed by the eavesdropper is independent with the input message in the secure communication problem, and hence the secrecy conditions are satisfied. Details are given in \cite{dikaliotis12}.

\ifCLASSOPTIONcaptionsoff
  \newpage
\fi



%
\bibliographystyle{IEEEtran}	
\bibliography{sec}

%




\end{document}